\documentclass[12pt, onecolumn, draftcls]{IEEEtran}

\usepackage[latin1]{inputenc}
\usepackage{graphicx}
\usepackage{amssymb,amsmath,amsfonts}
\usepackage{algorithm,algorithmic}
\usepackage{cite}
\usepackage{float}
\usepackage{graphics}
\usepackage{subfigure}
\usepackage{xspace}
\usepackage[usenames,dvipsnames]{color}
\usepackage{amssymb, epsfig}

\usepackage{amsfonts}
\usepackage{algorithm}
\usepackage{algorithmic}
\usepackage{cite}
\usepackage{float}
\usepackage{bbm}
\usepackage{subeqnarray}
\usepackage[latin1]{inputenc}
\usepackage{amsmath}
\usepackage{amssymb}
\usepackage{amsfonts}
\usepackage{algorithm}
\usepackage{algorithmic}
\usepackage{cite}
\usepackage{float}
\usepackage{graphics}
\usepackage{subfigure}
\usepackage{color}
\usepackage{epsfig}
\usepackage{bbm}
\usepackage{subeqnarray}
\makeatother

\newcommand{\qh}{{\bf h}}

\newcommand{\qp}{{\bf p}}

\newcommand{\qv}{{\bf v}}
\newcommand{\qw}{{\bf w}}

\newcommand{\qG}{{\bf G}}

\newcommand{\qI}{{\bf I}}

\newcommand{\qW}{{\bf W}}

\newcommand{\qzero}{{\bf 0}}

\newcommand{\tr}{\mbox{trace}}

\newtheorem{theorem}{Theorem}

\newcommand{\be}{\begin{equation}} \newcommand{\ee}{\end{equation}}
\newcommand{\bea}{\begin{eqnarray}} \newcommand{\eea}{\end{eqnarray}}
\newcommand{\benum}{\begin{enumerate}} \newcommand{\eenum}{\end{enumerate}}

\begin{document}

\title{Beamforming for MISO Interference Channels with QoS and RF Energy Transfer}

\author{Stelios Timotheou,~\IEEEmembership{Member,~IEEE,} Ioannis Krikidis,~\IEEEmembership{Senior Member,~IEEE,}~\\Gan Zheng,~\IEEEmembership{Senior Member,~IEEE,} and ~Bj\"orn Ottersten,~\IEEEmembership{Fellow,~IEEE}
\thanks{S. Timotheou and I. Krikidis are with the Department of Electrical and Computer Engineering, University of Cyprus, Cyprus (E-mail: \{timotheou.stelios, krikidis\}@ucy.ac.cy).}
\thanks{G. Zheng is with the School of Computer Science and Electronic Engineering, University of Essex, UK, (E-mail: ganzheng@essex.ac.uk). He is also affiliated with Interdisciplinary Centre for Security, Reliability and Trust (SnT),  University of Luxembourg, Luxembourg.}
\thanks{B. Ottersten is with the Interdisciplinary Centre for Security, Reliability and Trust (SnT), University of Luxembourg (E-mail: bjorn.ottersten@uni.lu).}
\thanks{ A preliminary version of this work that introduced the considered problem and discussed ZF and MRT fixed beamforming solutions appears in \cite{TIM}.}}

\maketitle

\maketitle

\begin{abstract}
We consider a multiuser multiple-input single-output interference channel where the receivers are characterized by both quality-of-service (QoS) and radio-frequency (RF) energy harvesting (EH) constraints. We consider the power splitting RF-EH technique where each receiver divides the received signal into two parts a) for information decoding and b) for battery charging. The minimum required power that supports both the QoS and the RF-EH constraints is formulated as an optimization problem that incorporates the transmitted power and the beamforming design at each transmitter as well as the power splitting ratio at each receiver.  We consider both the cases of fixed beamforming and when the beamforming design is incorporated into the optimization problem.
For fixed beamforming we study three standard beamforming schemes, the zero-forcing (ZF), the regularized zero-forcing (RZF) and the maximum ratio transmission (MRT); a hybrid scheme, MRT-ZF, comprised of a linear combination of MRT and ZF beamforming is also examined. The optimal solution for ZF beamforming is derived in closed-form, while optimization algorithms based on second-order cone programming are developed for MRT, RZF and MRT-ZF beamforming to solve the problem. In addition, the joint-optimization of beamforming and power allocation is studied using semidefinite programming (SDP) with the aid of rank relaxation.    
\end{abstract}

\begin{keywords}
Radio-frequency energy transfer, beamforming, power splitting, MISO channel, second-order cone programming, optimization.
\end{keywords}

\section{Introduction}

Recently, energy harvesting (EH) communication systems that can harvest energy from a variety of natural and man-made sources (solar, wind, mechanical vibration, etc.) for sustainable network operation has attracted much interest. EH opens new challenges in the analysis and design of transmission schemes and protocols that efficiently handle the harvested  energy.  Most literature concerns the optimization of different network utility functions under various assumptions on the knowledge of the energy profiles. The works in \cite{YEN,IAN,OZE,GUN} assume that the EH profile is perfectly known at the transmitters and investigate optimal resource allocation techniques for different objective functions and network topologies. On the other hand, the works in \cite{EPH,PAP} adopt a more networking point of view and maximize the stability region of the network by assuming only statistical knowledge of the EH profile.   However, the main limitation of the conventional EH sources is that in most cases they are not controlled and thus not always available; this uncertainty can be critical for some applications where reliability is of paramount importance.  A promising harvesting technology that could overcome this bottleneck, is the radio frequency (RF) energy transfer where the ambient RF radiation  is captured by the receiver antennas and  converted into a direct current (DC) voltage through appropriate circuits (rectennas) \cite{POP2}.  The RF-EH can be fully-controlled and therefore can be used for applications with critical quality-of-service (QoS) constraints. In addition, for some applications where other EH technologies can not be deployed (such as wireless soil sensor networks), RF energy transfer seems to be a suitable solution.

RF-EH is a new research area that attracts the attention of both academia and industry. The long term perspective is to be able to capture the electromagnetic radiation that is available in the surrounding (TV towers, cellular base-stations etc.) and use it in order to power communication systems. Most of the work on RF energy transfer focuses on the circuit and rectenna design as it is a vital requirement to make RF-EH feasible \cite{SUN,NIN}. On the other hand, the development of protocols and transmission techniques for wireless networks with RF-EH capabilities is a new research direction and few studies appear in the literature. The fundamental concept of simultaneous wireless transmission of energy and information is discussed in \cite{GROV} from an information theoretic standpoint. The work in \cite{SIM1} discusses the simultaneous information and energy transfer for a basic resource-constrained two-way communication link without energy losses (ideal energy recycling).  In \cite{FOU},  the authors characterize the capacity of two basic multi-user network configurations (multiple access channel, multi-hop channel) when simultaneous information and energy transfer is employed.

However, due to practical hardware constraints, simultaneous energy and information transmission is not possible with existing technology. In \cite{ZHA} the authors study practical beamforming techniques in a simplified multiple-input multiple-output (MIMO) network that ensure QoS and EH constraints for two separated receivers, respectively. In that work, the authors also discuss two practical receiver approaches for simultaneous wireless power and information transfer a) ``time switching'' (TS), where the receiver switches between decoding information and harvesting energy and b) ``power splitting'' (PS), where the receiver splits the received signal in two parts for decoding information and harvesting energy, respectively. This work is extended in \cite{XIA} for scenarios with imperfect channel state information (CSI) at the transmitter and a robust beamforming design is presented. The authors in \cite{LIU} deal with the TS technique and propose a dynamic switching between decoding information and EH in order to achieve various trade-offs between wireless information transfer and EH; both cases with and without channel state information (CSI) at the transmitter are considered. This fundamental switching between information transfer and EH is discussed in \cite{KRI} for a cooperative relaying scenario with a discrete-level battery at the relay node and no CSI at the transmitter. Other approaches implement this ``simultaneous'' information and energy transfer by assuming multiple receivers where  some of them use the transmitted signal for information decoding and other for energy harvesting i.e. \cite{XU, HIM}. In \cite{XU}, the authors investigate the optimal multiple-input single-output (MISO) beamformer for a network with a single transmitter and multiple information receivers and energy harvesters.  The work in \cite{HIM} deals with the problem of relay selection for a system where the selected relay conveys information to a data decoder while it simultaneously transfers energy to an energy harvester.

On the other hand, more recent works deal with the PS technique that allows an artificial simultaneous information and energy transfer \cite{NAS,LIU2}. In \cite{NAS}, the authors derive the outage probability and the ergodic capacity for a basic cooperative network with a PS-based relay; in that work the optimal power splitting ratio is evaluated based on numerical results. The work in \cite{LIU2} studies the optimal power splitting rule at the receiver in order to achieve various trade-offs between information transfer (ergodic capacity) and maximum average harvested energy; both cases with and without CSI at the transmitter are discussed. However, most of the work in the literature focuses on simple network topologies with single transmitters; the application of the RF-EH technology to more complex networks configurations is still an open problem.

In this paper, we focus on the PS approach \cite{ZHA} and we study a MISO interference channel where the downlink receivers are characterized by both QoS and EH constraints. For fixed beamforming weights and known CSI at the transmitters, we optimize the values for the power split at each receiver as well as the transmitted power for each source with the objective of minimizing the total transmitted power. Different solution methodologies are proposed for general beamforming schemes, such as the maximum ratio transmission (MRT) \cite{LAR,LEI}, zero-forcing (ZF) and reguralized zero-forcing (RZF) beamforming, and a hybrid scheme that combines the MRT and ZF beamformers (MRT-ZF),  while an optimal closed-form solution is derived for zero-forcing (ZF) beamforming. A comparison between them shows that ZF, which is considered as an efficient beamforming design for conventional MISO systems, requires significantly more transmit power compared to MRT beamforming, but always leads to feasible solutions for the considered problem (when the number of antennas at each transmitter is no less than the number of receivers). By combining the best of both worlds, MRT-ZF always provides feasible solutions of significantly better quality compared to the standard beamforming approaches, at a small increase of computational complexity. 

In addition to the case of fixed beamforming that keeps the complexity low, the optimal benchmark scheme that jointly optimizes beamforming weights, transmit power and power splitting ratios in order to minimize the total power consumption, is investigated. The use of semidefinite programming (SDP) together with rank relaxation is proposed to approximate the optimal solution and an algorithm to recover the beamforming solution is designed. Furthermore, we prove that when there are two or three users, the proposed approach always gives rank-1 solutions, from which the exactly optimal beamforming solution can be obtained. The SDP approach is superior in terms of solution quality amongst all the schemes investigated providing the optimal solution in all problem instances considered, but its execution time becomes prohibitive for increasing problem sizes; the best trade-off between solution quality and computational complexity is provided by the MRT-ZF beamforming scheme.

The contributions of this paper are the following:
\begin{enumerate}
	\item Introduction and formulation of the problem; MISO interference channel with simultaneous QoS and RF-EH requirements. 
	\item Efficient approximation of the optimal solution to the general problem using SDP; in practice, the SDP scheme provided the optimal solution in all problem instances examined. 
	\item Formulation of a second-order cone program for the optimal solution of the power allocation and splitting problem for any beamformer (constant beamforming weights). Derivation of a closed-form solution for ZF beamforming.
	\item Development of a fast and efficient suboptimal approach for the solution of the general problem that is based on optimally combining the MRT and ZF beamformers. 
\end{enumerate}

{\underline{\it Notation}:} All boldface letters indicate vectors (lower case) or matrices (upper case). The superscripts $(\cdot)^{T}$, $(\cdot)^H$,  $(\cdot)^{-1}$, $(\cdot)^{\dag}$ denote the transpose, the conjugate transpose, the matrix inverse and the Moore-Penrose pseudo-inverse respectively. A circularly symmetric complex Gaussian random variable $z$ with mean $\mu$ and variance $\sigma^2$ is represented as $\mathcal{CN}\left(\mu, \sigma^2\right)$. The identity matrix of size $M$, and the zero matrix of size $m\times n$, are denoted by $\mathbf{I}_{M}$ and $\mathbf{0}_{m\times n}$, respectively. $||\mathbf{z}||$ denotes the Euclidean norm of a complex vector $\mathbf{z}$, $|z|$ denotes the magnitude of a complex variable $z$, $\tr{(\mathbf{A})}$ denotes the trace of a matrix $\mathbf{A}$, while $\mathbf{A}\succeq 0$ indicates that matrix $\mathbf{A}$ is positive semidefinite.

This paper is organized as follows: Section II sets up the system model and introduces the optimization problem. Section III, investigates the joint optimization of the transmitted power, beamforming design and power splitting ratios.  Section IV discusses some fixed MISO beamforming schemes, while Section V investigates the solution of the arising optimization problems under these fixed beamforming schemes.  Simulation results are presented in Section VI, followed by our conclusions in Section VII.

\section{System model and problem formulation}\label{system_model}

We assume a MISO interference channel consisting of $K$ multiple-antenna sources and $K$ single-antenna receivers that employ single-user detection; each source $S_i$ communicates with its corresponding receiver $D_i$ ($i=1,\ldots,K$). This configuration moves almost all the signal processing at the transmitters (i.e., base-stations in a cellular context) and ensures simple receivers, which due to size limitation cannot support multiple antennas. Each transmitter requires that the number of transmit antennas is at least equal to the sum of all receivers' antennas in order to satisfy the dimensionality constraint required to force to zero the cross-interference at each receiver \cite{LAR,SPE}.  For the sake of simplicity and without loss of generality, we assume a symmetric topology where each source holds the minimum required number of antennas, which is equal to the number of the receivers in the considered setup (e.g., $K$ antennas at each transmitter). The case of having more receivers than transmit antennas is not considered  because the feasibility of different beamforming schemes is not guaranteed and more importantly, it is difficult to provide satisfactory QoS to individual receivers. The contribution of this work as well as our main conclusions are not restricted by this assumption.  The system model considered is depicted in Fig. \ref{model}. We assume that the source $S_i$ transmits with a power $P_i$ and let $s_i$ be its transmitted data symbol with $\mathbb{E}\{\|s_i\|^2 \}=1$. The transmitted data symbol $s_i$ is mapped onto the antenna array elements by the beamforming vector $\mathbf{w}_i \in \mathbb{C}^{K\times 1}$ with $\|\mathbf{w}_i\|=1$.

All wireless links exhibit independent fading and Additive White Gaussian Noise (AWGN) with zero mean and variance $\sigma^2$. The fading is assumed to be frequency non-selective Rayleigh block fading.  This means that the fading coefficients in the vector channel $\mathbf{h}_{i,j} \in \mathbb{C}^{K\times 1}$ (for the $S_j\rightarrow D_i$ link) remain constant during one slot, but change independently from one slot to another according to the distribution $\mathcal{CN}\left(0, \sigma_{ch_{i,j}}^2\right)$.  The variance of the channel coefficients captures large-scale degradation effects such as path-loss and shadowing. The received symbol sampled baseband signal at receiver $D_i$ can be expressed as
\begin{align}\label{sys1}
y_i=\underbrace{\sqrt{P_i}\mathbf{h}_{i,i}^T \mathbf{w}_i s_i}_{\textrm{Information signal}}+\underbrace{\sum_{j\neq i}\sqrt{P_j}\mathbf{h}_{i,j}^T \mathbf{w}_js_j}_{\textrm{Interference}}+n_i,
\end{align}
where $n_i$ denotes the AWGN component; therefore the received power at $D_i$ is equal to
\begin{equation}\label{eq:TotalPower}
P^r_i=\displaystyle \sum_{j=1}^{K} |\mathbf{h}_{i,j}^T \mathbf{w}_j|^2 P_{j} + \sigma^2.
\end{equation}

The receivers have RF-EH capabilities and therefore can harvest energy from the received RF signal based on the power splitting technique \cite{ZHA}. With this approach, each receiver splits its received signal in two parts a) one part is converted to a baseband signal for further signal processing and data detection and b) the other part is driven to the required circuits for conversion to DC voltage and energy storage. Let $\rho_i \in (0,1)$ denote the power splitting parameter for the $i$-th receiver; this means that $100 \rho_i \%$ of the received power is used for data detection while the remaining amount is the input to the RF-EH circuitry.  More specifically, after reception of the RF signal at the receiver, a power splitter divides the power $P^r_i$ into two parts according to $\rho_i$, so that $\rho_i P^r_i$ is directed towards the decoding unit and $(1 - \rho_i)P^r_i$ towards the EH unit. During the baseband conversion, additional circuit noise, $v_i$, is present due to phase offsets and non-linearities which is modeled as AWGN with zero mean and variance $\sigma_C^2$ \cite{ZHA}. Fig. \ref{power_split} schematically shows the power splitting technique for the $i$-th receiver.

The signal-to-interference-plus-noise ratio (SINR) metric characterizing the data detection process at the $i$-th receiver is given by:
\begin{equation}
\Gamma_i= \frac{\rho_i P_i|\mathbf{h}_{i,i}^T \mathbf{w}_i|^2}{\rho_i\left(\sigma^2+\sum_{j\neq i}P_j|\mathbf{h}_{i,j}^T\mathbf{w}_j|^2\right) + \sigma_C^2}.
\label{eq:Gamma}
\end{equation}

On the other hand, the total electrical power that can be stored is equal to $P_i^S = \zeta_i (1-\rho_i) P^r_i$, where $\zeta_i \in (0,1]$ denotes the conversion efficiency of the $i$-th EH unit\footnote{The parameter $\zeta$ depends on the frequency of operation, the received RF energy as well as the specifications of the diode-based rectification circuit. In \cite{POP2} a practical cellular energy harvesting system was designed,  operating at $\zeta_i\approx 0.25$ conversion efficiency.} ($100\cdot \zeta_i$\% of the RF energy received at the EH unit can be stored as electrical energy). We note that in this study, the RF-EH constraints considered refer to the required rectennas' input without discussing energy storage efficiency issues.

\begin{figure}[t]
\centering
\includegraphics[width=0.6\linewidth]{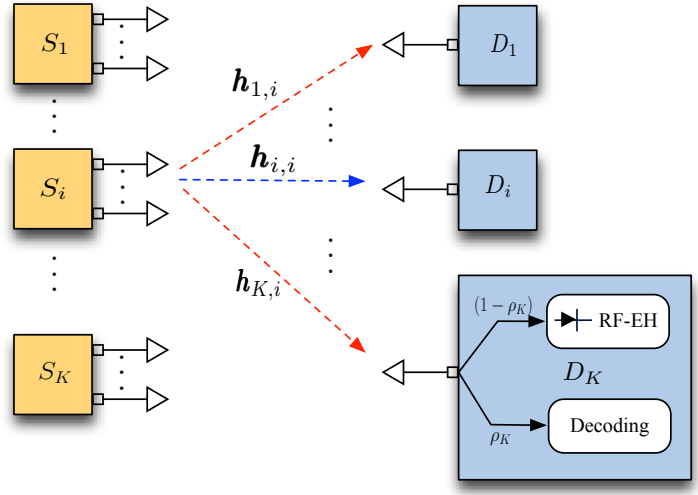}
\vspace{-0.5cm}
\caption{The MISO system model; each receiver splits the received energy in two parts for decoding and EH, respectively.}\label{model}
\end{figure}

\begin{figure}[t]
\centering
\includegraphics[width=0.6\linewidth]{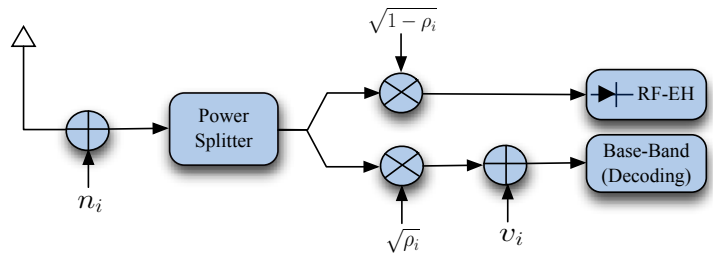}
\vspace{-0.5cm}
\caption{The power splitting technique for the $i$-th receiver.}\label{power_split}
\end{figure}
\vspace{-0.5cm}

\subsection{Optimization problem}

We assume that both receivers are characterized by strict QoS and EH constraints. The QoS constraint requires that the SINR should be higher than the threshold $\gamma_i$; the energy constraint requires that the input to the RF energy circuitry is higher than the energy threshold $\lambda_i$.   The energy harvesting constraint represents the minimum amount of energy that should trigger the rectenna's input of each receiver in order to ensure an efficient amount of energy harvested in each transmission time. This energy constraint and the associated energy harvested maintains operability and fully powers the equipment for generating/decoding signals, the power amplifiers, the antennas as well as all the devices that are involved in a radio transmission/reception process;  the harvested energy can be accumulated for future use (batteries) or used immediately. It is worth noting that relevant works such as \cite{ZHA,LIU} assume similar energy harvesting constraints (in terms of a target average harvested energy) in a different context. We focus on beamforming design, as well as power allocation and splitting, in order to minimize the total transmitted power subject to QoS and energy harvesting constraints. Based on the previous notation, the optimization problem can be defined as
\begin{align}
\vspace{-0.3cm}
\label{initFormulation}
&\min_{\mathbf{P}, \mathbf{W},\boldsymbol \rho} \sum_{i=1}^K P_i, \nonumber \\
&\textrm{subject to} \left\{
\begin{array}{l} \Gamma_i \geq \gamma_i, ~\forall i\\
(1-\rho_i)P_i^r \geq \lambda_i,~\forall i\\
P_i\ge 0, \|\mathbf{w}_i\|^2  = 1,\;\; 0 \le \rho_i \le 1,~\forall i
\end{array} \right.
\end{align}
where $P_i^r$, and $\Gamma_i$ are defined in \eqref{eq:TotalPower} and \eqref{eq:Gamma} respectively, while $\mathbf{W}=[\mathbf{w}_1,...,\mathbf{w}_K]$. 

The problem can be rewritten into an equivalent form by utilizing unnormalized beamforming weights $\qv_i=\sqrt{P_i}\qw_i$ which combine beamforming and power allocation as:
\bea\label{eqn:prob:0}
\min_{\{\qv_i, \rho_i\}}&& \sum_{i=1}^K \|\qv_i\|^2  \\
\mbox{s.t.}&&
  \Gamma_i = \frac{ |\qh_{i,i}^T\qv_i|^2}{\sum\limits_{j \ne
i}  |\qh_{i,j}^T\qv_j|^2 + \sigma^2 + \frac{\sigma_C^2}{\rho_i}} \ge \gamma_i,\notag\\
&& (1-\rho_i) \left(\sum\limits_{j=1}^K  |\qh_{i,j}^T\qv_j|^2 + \sigma^2  \right)\ge \lambda_i, 0\le \rho_i \le 1 ~\forall i.\notag \eea
It is easy to see that formulation \eqref{eqn:prob:0} is non-convex and hence quite challenging to solve. 

In the traditional beamforming problem the SINR constraints become convex by expressing them as second order rotated cone constraints \cite{SOCP_Beamforming}: 
\[\gamma_i \left(\sigma^2+\sum_{j\neq i}|\mathbf{h}_{i,j}^T\mathbf{v}_j|^2 \right) \le |\mathbf{h}_{i,i}^T \mathbf{v}_i|^2, ~\forall i\]
\[\mathbf{h}_{i,i}^T \mathbf{v}_i \ge 0,~\forall i,\]
Nonetheless, the introduction of the power splitting parameters destroys this useful structure and makes the SINR constraints non-convex. 

Even if we assume that the power splitting parameters are constant, the RF-EH constraints are non-convex because they are composed of a sum of concave terms:
\[ - (1-\rho_i)\sum_{j=1}^{K} |\mathbf{h}_{i,j}^T \mathbf{v}_j|^2 \le  (1-\rho_i)\sigma^2 - \lambda_i, \forall i.\]
Hence even if the power splitting parameters are constant in formulations \eqref{initFormulation} or \eqref{eqn:prob:0}, the problem is a non-convex Quadratically Constrained Quadratic Program (QCQP) \cite{QCQP_NP_Hard}. Having variable power splitting parameters further increases the difficulty of dealing with the problem.

\section{Optimal beamforming}
\label{sec:optBeamforming}

In this section we develop an efficient SDP algorithm that jointly optimizes the beamforming vectors, as well as the power and power splitting parameters. Joint optimization provides optimal performance at the cost of relatively high computational complexity. Hence, it is more suitable for the solution of small problems or as a performance benchmark. 

\subsection{Semidefinite programming (SDP) with rank relaxation}
  To tackle (\ref{eqn:prob:0}), we first  introduce new matrix variables $\qW_i=\qv_i\qv_i^H, \forall i$. Using $\{\qW_i\}$, problem
  (\ref{eqn:prob:0}) can be rewritten as
 \begin{subeqnarray} \label{eqn:prob:sdp} \slabel{subeq:sdp0}
\min_{\{\qW_i, \rho_i,\,\forall i\}}&& \sum_{i=1}^K \tr(\qW_i)  \\
\mbox{s.t.}&&
    \frac{ \tr(\qh_{i,i}^*\qh_{i,i}^T\qW_i)}{\sum\limits_{j=1,j \ne
i}^K  \tr(\qh_{i,j}^*\qh_{i,j}^T\qW_j) + \sigma^2 + \frac{\sigma_C^2}{\rho_i}} \ge \gamma_i,\slabel{subeq:sdp1}\\
&& (1-\rho_i) \left(\sum\limits_{j=1}^K  \tr(\qh_{i,j}^*\qh_{i,j}^T\qW_j)  + \sigma^2  \right)\ge \lambda_i ~\forall i,\slabel{subeq:sdp2}\\
&& \mathbf{W}_{i}\succeq\mathbf{0},~0\le\rho_{i}\le1,\,\forall i  \slabel{subeq:sdp3}
\end{subeqnarray}
 which can be further expressed as
\begin{eqnarray}\label{eqn:prob:sdp2}
{\displaystyle \min_{\{\mathbf{W}_{i},\,\rho_{i},\,\forall i\}}} && {\displaystyle \sum_{i=1}^{K}}\text{trace}(\mathbf{W}_{i}) \\
\mbox{s.t.}&& \mathbf{A}_i=\left[\begin{array}{ll}
\frac{1}{\gamma_{i}}\text{trace}(\mathbf{h}_{i,i}^{*}\mathbf{h}_{i,i}^{T}\mathbf{W}_{i})-\sum\limits _{j=1,j\ne i}^{K}\text{trace}(\mathbf{h}_{i,j}^{*}\mathbf{h}_{i,j}^{T}\mathbf{W}_{j})-\sigma^{2}, & \sigma_{C}\\
\sigma_{C}, & \rho_{i}
\end{array}\right]\succeq0,\,\forall i \notag\\
 && \mathbf{B}_i=\left[\begin{array}{ll}
\sum\limits _{j=1}^{K}\text{trace}(\mathbf{h}_{i,j}^{*}\mathbf{h}_{i,j}^{T}\mathbf{W}_{j})+\sigma^{2}, & \sqrt{\lambda_{i}}\\
\sqrt{\lambda_{i}}, & 1-\rho_{i}
\end{array}\right]\succeq0,\,\forall i 	 \notag\\
 && \mathbf{W}_{i}\succeq\mathbf{0},~0\le\rho_{i}\le1,\,\forall i  \notag
\end{eqnarray}
In the above formulation, matrices $\mathbf{A}_i$ and $\mathbf{B}_i$, $\forall i$ are hermitian, while they are positive semidefinite if and only if constraints \eqref{subeq:sdp1} and \eqref{subeq:sdp2} are respectively true. The reason is that constraints \eqref{subeq:sdp1} and \eqref{subeq:sdp2} ensure that all leading principal minors of matrices $\mathbf{A}_i$ and $\mathbf{B}_i$ are nonnegative, a condition that ensures that a hermitian matrix is positive semidefinite.

Note that in order to make  (\ref{eqn:prob:sdp2}) equivalent to the original problem (\ref{eqn:prob:0}), additional rank constraints
$\mbox{rank}(\qW_i)=1$ should be added to (\ref{eqn:prob:sdp2}), which are nonconvex and difficult to deal with. To make the problem tractable, we
first relax this rank constraints and focus on (\ref{eqn:prob:sdp2}). This problem is convex and belongs to the class of SDP, as it is composed of a linear objective function over the intersection of cones of positive semidefinite matrices and linear matrix inequalities involving $\{\qW_i\}$ and $\rho_i$ variables. Its numerical solution can be found by using SDP solvers.

Once (\ref{eqn:prob:sdp2}) is optimally solved, if the resulting solutions $\{\qW_i\}$ are all rank-1, then they are the exactly optimal solutions;
otherwise, (\ref{eqn:prob:sdp2}) provides a lower bound for the minimum required power and how to extract a feasible solution will be discussed
later.  Note that formulation \eqref{eqn:prob:sdp2} can be easily extended to include transmit power constraints by adding the linear constraints $ \text{trace}(\mathbf{W}_{i}) \le P_i^{max}$, $\forall i$, as the problem remains convex and its complexity is not significantly affected.

 \subsection{Rank Issue and The Proposed Algorithm}
 As mentioned before, problem (\ref{eqn:prob:sdp2}) may not give the optimal solution to the original problem   (\ref{eqn:prob:0}) due to rank relaxation. In our simulation, we find that (\ref{eqn:prob:sdp2}) always gives rank-1 solutions which are also optimal to the original problem (\ref{eqn:prob:0}), however, this property is yet to be proven and left
 for future work.  In this section, we propose the following algorithm to find a good heuristic solution when higher-rank solutions are returned by (\ref{eqn:prob:sdp2}).

\underline{\textbf{ Proposed Algorithm 1:}}
 \begin{enumerate}
    \item Solve problem (\ref{eqn:prob:sdp2}) to obtain the optimal $\{\qW_i^\star\}$.
    \item For each $i$, if $\mbox{rank}(\qW_i)=1$, find eigenvalue decomposition $\qw_i^\star$ such that $\qW_i^\star = p_i\qw_i^\star{\qw_i^\star}^H$, $\|\qw_i^\star\|=1$; otherwise, choose
    $\qw_i^\star$ as the principal eigenvector of $\qW_i$, which corresponds to the largest eigenvalue of $\qW_i$.
    \item  With $\{\qw_i^\star\}$, solve the power allocation problem with fixed weights (\ref{eq:SOCP}) to obtain $\qp^\star$.
    \item Return the optimal beamforming vectors $\{\sqrt{p_i^\star}\qw_i^\star\}$.
 \end{enumerate}
 If $\mbox{rank}(\qW_i)>1$, we can also use randomization techniques to find $\qw_i$ \cite{Luo}.

 \subsection{Special Cases: K=2 and K=3}
 In general it is not known whether problem (\ref{eqn:prob:sdp2}) can guarantee to return rank-1 solutions which is optimal to (\ref{eqn:prob:0}). In
 this section, we study the rank properties of two special cases:  $K=2$ and $K=3$. We first give the following results.
 \begin{theorem}
    When $K=2$ and $K=3$, there exist rank-1 solutions that optimally solves problem (\ref{eqn:prob:sdp2}), which are also optimal to the original problem
    (\ref{eqn:prob:0}).
 \end{theorem}
 \begin{proof}
    Suppose that $\{\qW_i^\star\}$ are the optimal solutions to (\ref{eqn:prob:sdp2}) and $\mbox{rank}(\qW_i^\star)\ge 1$ because $\qW_i^\star\ne \qzero$, then according to \cite[(24)]{Rank-Huang},
    we have the following results about the rank of the solutions to (\ref{eqn:prob:sdp2}) with $K$ QoS and EH constraints
    for a general $K$-user MIMO interference channel:
    \be \label{eqn:rank}
        K \le \sum_{i=1}^K (\mbox{rank}(\qW_i^\star))^2\le 2K.
    \ee
    When $K=2$, (\ref{eqn:rank}) becomes
    \be
       2 \le  (\mbox{rank}(\qW_1^\star))^2 + (\mbox{rank}(\qW_2^\star))^2\le 4,
    \ee
    from which we can derive that $\mbox{rank}(\qW_1^\star)=\mbox{rank}(\qW_2^\star)=1$, thus problem (\ref{eqn:prob:sdp2}) always gives rank-1 solutions when
    $K=2$.

    When $K=3$,  (\ref{eqn:rank}) is expanded as
       \be
       3 \le  (\mbox{rank}(\qW_1^\star))^2 + (\mbox{rank}(\qW_2^\star))^2+ (\mbox{rank}(\qW_3^\star))^2\le 6,
    \ee
   which shows that there is at most one $\qW_i^\star$ that is rank-2 and the other two $\qW_j^\star$s, $j\ne i$, must be rank-1. Without loss of generality, we
   assume that $\mbox{rank}(\qW_1^\star)=2, \mbox{rank}(\qW_2^\star)=1, \mbox{rank}(\qW_3^\star)=1$. Next we show that there exists $\qv_1^\star$ such that $\qv_1^\star{\qv_1^\star}^H$ is also an optimal solution to
   (\ref{eqn:prob:sdp2}). To this end, it suffices to prove that  there exists $\qv_1^\star$ that satisfies the following equations:
\begin{align}
&\tr(\qW_1^\star) = \|\qv_1^\star\|^2 \\
&\qh_{j,1}^T \qW_1^\star \qh_{1,j}^*  =   |\qh_{1,j}^T\qv_1^\star|^2,\;\;\;\;\;\text{for}\;j=1,2,3. \notag  
\end{align}

   It is proved in \cite[Theorem 2.3]{Decomposition-Huang} that the above decomposition indeed exists. Furthermore, Algorithm 3 in
   \cite{Decomposition-Huang} provides the detailed procedures about how to find $\qv_1^\star$ that satisfies the above equations. This completes the proof.
 \end{proof}

\section{Conventional MISO beamforming schemes}
\label{sec:beamformers}

Although, the developed SDP algorithm may provide optimal solutions, its high computational complexity make it unsuitable for providing real-time solutions. In this section, we briefly present standard beamforming designs that are well-known in conventional MISO systems that facilitate the development of low-complexity algorithms.  More specifically, we focus on ZF and MRT beamforming schemes, which represent extreme situations towards achieving the SINR and RF-EH constraints, respectively. We also study two beamforming schemes that attempt to balance the two extremes: RZF which does not fully cancel cross-interference,
and MRT-ZF, a hybrid scheme that considers a linear combination of ZF and MRT.

\subsection{Zero Forcing}
\label{sec:ZF}

In ZF, the weights are selected such that the co-channel interference
is canceled, i.e. for the desired user $i$ the ZF condition becomes
$\mathbf{h}_{i,j}^{T}\mathbf{w}_{j}=0$, $j\neq i$. More specifically, the ZF weights
are computed as the solution of the following optimization problem:
\begin{eqnarray*}
\max_{\mathbf{w}_{i}}|\mathbf{h}_{i,i}^{T}\mathbf{w}_{i}|^{2}\nonumber \\
\textrm{s. t.}\;\;\mathbf{H}_{i}^{T}\mathbf{w}_{i}=\mathbf{0}_{(K-1)\times1}\nonumber\\
 \;\;\;\;\;\; \|\mathbf{w}_{i}\|^2=1
\end{eqnarray*}
where $\mathbf{H}_{i}=[\mathbf{h}_{1,i},\ldots,\mathbf{h}_{i-1,i},\mathbf{h}_{i+1,i},\ldots,\mathbf{h}_{K,i}]^{T}$.
The solution of the beamforming is given by \cite{LEI}
\noindent
\begin{align}
\mathbf{w}_{i}^{\textrm{(ZF)}}=\frac{(\mathbf{I}_K-\mathbf{F}_i)\mathbf{h}_{i,i}^{*}}{\|(\mathbf{I}_K-\mathbf{F}_i)\mathbf{h}_{i,i}^{*}\|}
\end{align}
where $\mathbf{F}_i=\mathbf{H}_{i}^{\dagger}\mathbf{H}_{i}$ and $\mathbf{H}_{i}^{\dagger}=\mathbf{H}_{i}^{H}(\mathbf{H}_{i}\mathbf{H}_{i}^{H})^{-1}$
is the Moore-Penrose inverse of $\mathbf{H}_{i}$.

By fully canceling cross-interference, on the one hand the ZF beamformer optimizes the SINR constraints, but on the other hand it puts little emphasis on achieving the RF-EH constraints.

\subsection{Maximum Ratio Transmission}
\label{sec:MRT}

The MRT beamforming maximizes the signal-to-noise ratio (SNR) at each receiver ($|\mathbf{h}_{i,i}^T\mathbf{w}_{i}|^2/\sigma^2$) and requires only the knowledge of the direct links $\mathbf{h}_{i,i}$; due to this limited CSI knowledge, it is of low complexity and is suitable for practical applications with strict computational/time constraints.  The MRT beamforming is expressed as \cite{LAR}
\begin{align}
\mathbf{w}^{\textrm{(MRT)}}_i=\frac{\mathbf{h}_{i,i}^*}{\|\mathbf{h}_{i,i}^*\|}.
\end{align}
It is worth noting that the MRT does not take into account the simultaneous sources' transmissions and therefore it results in a strong cross-interference.  Although this cross-interference is a bottleneck for conventional MISO systems, it could be beneficial for scenarios with RF-EH constraints.

From the description of the beamforming schemes presented in Sections \ref{sec:ZF}
and \ref{sec:MRT}, it is clear that different beamformers result in different trade-offs
between the SINR and the RF-EH constraints. In the following, we consider two beamforming schemes that aim at balancing these two constraints.

\subsection{Regularized Zero Forcing (RZF)}
\label{sec:RZF}
The ZF beamforming scheme can be problematic when the channel is ill-conditioned; to remedy this limitation, RZF is proposed in \cite{RZF}, which aims to deal with this problem
as well as take into account the noise variance. The RZF beamforming vector for user $i$ admits the following expression
\begin{align}
\mathbf{w}_{i}^{\textrm{(RZF)}}=\frac{\left(\qG_i\qG^\dag_i + \eta_i \qI_K\right)^{-1}\mathbf{h}_{i,i}^{*}}{\|\left(\qG_i\qG^\dag_i + \eta_i
\qI_K\right)^{-1}\mathbf{h}_{i,i}^{*}\|},
\end{align}
where $\mathbf{G}_{i}=[\mathbf{h}_{1,i},\ldots,\mathbf{h}_{K,i}]^{T}$ and $\eta_i$ is the regularization parameter. It can be seen that RZF
beamforming does not completely cancel co-channel interference while $\eta_i$ controls interference to user $i$. Notice that RZF provides a better
trade-off between useful signal and interference than ZF because it allows controlled interference which is desired for RF-EH constraints. Ideally
$\eta_i$ should be optimized, however, this is a nontrivial task \cite{RZF2} which requires the consideration of many factors like channel fading
characteristics, antenna correlation, design objectives, etc. In this paper, to keep the complexity low, we use the same constant regularization
parameters for all users, i.e., $\eta_i=c,\forall i$.

\subsection{Hybrid Maximum Ratio Transmission - Zero-Forcing beamforming}

Another potentially good solution could result from the linear combination of
MRT and ZF into a hybrid MRT-ZF scheme, which attempts to find the
best trade-off of the two. A linear combination of MRT and ZF beamformers was shown to provide a complete description of the Pareto boundary for the achievable rate of MISO interference channel for the two-user case \cite{LAR}. The purpose in this section is to combine the two beamformers to achieve a good trade-off between the SINR and the RF-EH constraints. The MRT-ZF beamforming can be expressed as:
\begin{equation}
\mathbf{w}_{i}^{\text{(MRT-ZF)}}=\frac{\sqrt{x_{i}}\mathbf{r}_{i}^{\text{(MRT)}}+\sqrt{y_{i}}\mathbf{r}_{i}^{\text{(ZF)}}}{\|\sqrt{x_{i}}\mathbf{r}_{i}^{\text{(MRT)}}+\sqrt{y_{i}}\mathbf{r}_{i}^{\text{(ZF)}}\|},
\label{eq:weightsMRTZF}
\end{equation}
\noindent where $\mathbf{r}_{i}^{\text{(MRT)}}=\mathbf{h}_{i,i}^{*}$, $\mathbf{r}_{i}^{\text{(ZF)}}=(\mathbf{I}_K-\mathbf{F}_i)\mathbf{h}_{i,i}^{*}$, $\mathbf{F}_i$ as defined in ZF beamforming, while $x_{i},y_{i}\ge0$ are decision variables that need to be chosen to achieve an optimal trade-off between the two beamforming schemes. Note that for $y_i=0$ ($x_i=0$) this beamformer is equivalent to the MRT (ZF) beamformer.

\section{Problem solution for fixed beamforming}

In this section we propose solutions to the considered problem for the beamforming schemes described in Section \ref{sec:beamformers}.
In ZF beamforming, the optimization problem attains a special form that allows an optimal closed-form solution. In MRT and RZF beamforming, the optimization problem does not have any special form, but the problem can be transformed into a second-order cone programming (SOCP) formulation which leads to optimal, robust and fast solutions using off-the-shelf optimization solvers. The proposed solution method applies to any arbitrary beamforming scheme with fixed weights $\mathbf{w}_i^f$. In MRT-ZF beamforming, the optimal contribution of ZF and MRT beamforming has to be found; although the resulting problem is not convex, we develop an approximate solution based on SOCP which attains excellent results in practice.

\subsection{Problem solution for ZF beamforming}
\label{sec:ZFsol}

Letting $G_{i,j}=|\mathbf{h}_{i,j}^{T}\mathbf{w}_{j}|^2$
denote the link gain between $S_{i}$ and $D_{j}$, the SINR and the total received power at
the $i$-th receiver are simplified because $G_{i,j}=0$, for $i \neq j$. Hence the original optimization problem
simplifies into the following formulation:
\begin{subeqnarray} \label{eq:ZF} \slabel{eq:ZF1}
{\cal ZF}:~~~ \displaystyle{\min_{\mathbf{P}, \mbox{\boldmath $\rho$}}} && \sum_{i=1}^{K}P_{i}\\
\slabel{eq:ZF2} \textrm{s.t.} \frac{\rho_{i}G_{i,i} P_{i}}{\rho_i \sigma^2 + \sigma_C^2} & \geq & \gamma_{i}, \forall i\\
\slabel{eq:ZF3}(1-\rho_{i})(G_{i,i}P_i+\sigma^2) & \geq & \lambda_{i}, \forall i\\
\slabel{eq:ZF4} P_{i} \ge 0,&&0 \le \rho_{i} \le 1, \forall i.
\vspace{-0.5cm}
\end{subeqnarray}

Its solution is given in Theorem 2 below.
\begin{theorem}
Let $\alpha_i=(\gamma_{i}+1)\sigma^{2}$ and $\beta_i=\gamma_{i}\sigma_{C}^{2}$.
A feasible solution to optimization problem ${\cal ZF}$ always exists,
while its optimal solution can be expressed in closed-form as:
\end{theorem}
\begin{align*}
P_{i}^{*} &=\frac{1}{G_{i,i}}\left(\frac{\alpha_i+\beta_i+\lambda_{i}+\sqrt{(\alpha_i+\beta_i+\lambda_{i})^{2}-4\lambda_{i}\alpha_i}}{2}-\sigma^{2}\right),\\
\rho_{i}^{*} &= \frac{\beta_i}{G_{i,i}P_{i}^{*}+\sigma^{2}-\alpha_i}=1-\frac{\lambda_{i}}{G_{i,i}P_{i}^{*}+\sigma^{2}}.
\vspace{-0.5cm}
\end{align*}
\begin{proof} The proof of {\it Theorem 2} is provided in Appendix A. \end{proof}
Theorem  is very important as it demonstrates that there is always a feasible solution to optimization problem \eqref{initFormulation}, despite the presence of QoS and RF-EH constraints, no matter how demanding these constraints are. 

\subsection{Problem solution for MRT and RZF beamforming}

Contrary to the special case of ZF beamforming where $G_{i,j} \neq 0$, $\forall~i=j$, that allowed the decomposition of the problem, in MRT or RZF beamforming we have that $G_{i,j} \neq 0$, $\forall i,~j$. Hence, although  MRT or RZF have different link gain values, they share the same problem structure which results in the same formulation. Rearranging the terms in formulation \eqref{initFormulation} yields:
\begin{subeqnarray} \label{eq:OFW} \slabel{eq:OFW1}
\vspace{-0.5cm}
{\cal FW}:~~~ \displaystyle{\min_{\mathbf{P}, \mbox{\boldmath $\rho$}}} && \sum_{i=1}^{K}P_{i}\\
\slabel{eq:OFW2} \textrm{s.t.} (1+\gamma_{i}) G_{i,i} P_{i} & \geq & \frac{1}{\rho_i}\gamma_{i}\sigma_C^{2} + \gamma_{i} P^r_{i}, \forall i\\
\slabel{eq:OFW3} (1-\rho_{i})P^r_{i} & \geq & \lambda_{i}, \forall i\\
\slabel{eq:OFW4}{\displaystyle \sum_{j=1}^{K}}G_{i,j}P_{j} +\sigma^2 & = & P^r_{i}, \forall i \\
\slabel{eq:OFW5} P_{i} \ge 0,~&& 0 \le \rho_{i} \le 1, \forall i.
\end{subeqnarray}
Optimization problem $\cal{FW}$ is convex because it is comprised of a linear objective function and convex constraints. Constraint (\ref{eq:OFW2}) is convex because the term $\frac{1}{\rho_i}\gamma_{i}\sigma_C^{2}$ is convex for $\rho_i>0$ and the other terms are linear, (\ref{eq:OFW3}) is a restricted hyperbolic constraint and (\ref{eq:OFW4}) is linear.  Note that by solving problem $\cal{FW}$ we obtain optimal values for the splitting parameters, as well as an optimal power allocation for \emph{any} set of fixed beamforming vectors. Next we show how the problem can be cast into a SOCP formulation, which can be optimally and reliably solved using off-the-shelf algorithms.

\subsubsection*{SOCP Solution}

SOCP problems are convex optimization problems involving a linear function minimized subject to linear and second-order cone (SOC) constraints, which can be solved using fast and robust of-the-shelf solvers \cite{SOCP1}. Among the constraints that can be modeled using SOCs are the \textit{restricted hyperbolic constraints} which have the form: $\mathbf{x}^T \mathbf{x}\le yz$, which are equivalent to a rotated SOC constraint of the form:
\[
\left\Vert \left(\begin{array}{c}
2 \mathbf{x}\\
y-z
\end{array}\right)\right\Vert \le y+z, 
\]
where $\mathbf{x}\in \mathbb{C}^{N \times 1}$, $y,z \ge 0$. For example, constraint \eqref{eq:OFW3} is equivalent to the following SOC:
\[
\left\Vert \left(\begin{array}{c}
2 \sqrt{\lambda_i}\\
(1-\rho_i)-P^r_i
\end{array}\right)\right\Vert \le (1-\rho_i)+P^r_i.
\]
In order to cast problem $\cal{FW}$ into SOCP form we need to convert \eqref{eq:OFW2} into a restricted hyperbolic constraint. Let us define $z_i$=$(1+\gamma_{i}) G_{i,i} P_{i} - \gamma_{i} P^r_{i}$; it must hold  true that $z_i \ge 0$ since $\frac{1}{\rho_i}\gamma_{i}\sigma_C^{2}\ge 0$ otherwise \eqref{eq:OFW2} will be infeasible. Substituting $z_i$ into \eqref{eq:OFW2} yields that $\rho_i z_i\ge \gamma_{i}\sigma_C^{2}$ which is a convex restricted hyperbolic constraint. Hence, problem $\cal{FW}$ can be cast into the following convex problem which is equivalent to an SOCP formulation:
\begin{subeqnarray} \label{eq:SOCP} \slabel{eq:SOCP1}
\vspace{-0.5cm}
{\cal SOCP}:~~~ \displaystyle{\min_{\mathbf{P}, \mathbf{z}, \mbox{\boldmath $\rho$}}} && \sum_{i=1}^{K}P_{i}\\
\slabel{eq:SOCP2} \textrm{s.t.    } z_i \rho_i & \geq & \gamma_{i}\sigma_C^{2}, \forall i\\
\slabel{eq:SOCP3} z_i + \gamma_{i} P^r_{i}&=& (1+\gamma_{i}) G_{i,i} P_{i} , \forall i\\
\slabel{eq:SOCP4} (1-\rho_{i})P^r_{i} & \geq & \lambda_{i}, \forall i\\
\slabel{eq:SOCP5}{\displaystyle \sum_{j=1}^{K}}G_{i,j}P_{j} +\sigma^2 & = & P^r_{i}, \forall i \\
\slabel{eq:SOCP6} P_{i} \ge 0,~& &z_i\ge0,~0 \le \rho_{i} \le 1, \forall i.
\end{subeqnarray}
We should emphasize here that formulation \eqref{eq:SOCP} is general enough to optimize the power allocation and power splitting parameters for any beamforming scheme with fixed weights.  Formulation \eqref{eq:SOCP} can also be easily extended to include transmit power constraints by adding the bound constraints $ 0 \le P_{i} \le P_i^{max}$, $\forall i$.

\vspace{-0.3cm}
\subsection{Problem solution for MRT-ZF beamforming}

Assuming unnormalized weights $\mathbf{v}_i^{\text{(MRT-ZF)}}=\sqrt{P_i} \mathbf{w}_i^{\text{(MRT-ZF)}}$, the link gains $G_{i,j}$ resulting from the MRT-ZF beamformer are:
\begin{equation}
G_{i,j}=|\mathbf{h}_{i,j}^{T} \mathbf{v}_j^{\text{(MRT-ZF)}}|^{2}=\begin{cases}
x_{j}|\mathbf{h}_{i,j}^{T}\mathbf{h}_{j,j}^{*}|^{2}, & i\neq j\\
|\sqrt{x_{i}}\mathbf{h}_{i,i}^{T}\mathbf{h}_{i,i}^{*}+\sqrt{y_{i}}\mathbf{h}_{i,i}^{T}(\mathbf{I}_K-\mathbf{F}_i)\mathbf{h}_{i,i}^{*}|^{2}, & i=j.
\end{cases}\label{eq:chGainsMRTZF1}
\end{equation}

Substituting $Q_{i,j}=|\mathbf{h}_{i,j}^{T}\mathbf{h}_{j,j}^{*}|$,
$q_{i}=\mathbf{h}_{i,i}^{T}(\mathbf{I}_K-\mathbf{F}_i)\mathbf{h}_{i,i}^{*}\ge0$
and $s_{i}=\sqrt{x_{i}y_{i}}\ge0$ into \eqref{eq:chGainsMRTZF1} we obtain:
\begin{equation}
G_{i,j}=|\mathbf{h}_{i,j}^{T}\mathbf{v}_{j}|^{2}=\begin{cases}
x_{j}Q_{i,j}^{2}, & j\neq i\\
x_{i}Q_{i,i}^{2}+y_{i}q_{i}^{2}+2s_{i}Q_{i,i}q_{i}, & j=i.
\end{cases}\label{eq:chGainsMRTZF2}
\end{equation}

The above expression is true because for $i\neq j$, there is no contribution
to the link gains from the ZF component as $\mathbf{h}_{ij}^{T}\mathbf{r}_{i}^{\text{(ZF)}}=0$,
while for $i=j$ both terms in the norm are real and positive as $\mathbf{h}_{i,i}^{T}\mathbf{h}_{i,i}^{*}=|\mathbf{h}_{i,i}^{T}\mathbf{h}_{i,i}^{*}|=Q_{i,i}\ge0$
and $q_{i}\ge0$. The above equation indicates that the link gains $G_{i,j}$ can be written as a linear combination of variables $x_{i},\, y_{i},\, s_{i}$.
In addition, the power $P_{i}$ associated with a particular beamforming
vector $\mathbf{v}_i^{\text{(MRT-ZF)}}$ can be expressed as:
\begin{align}
P_{i} & =\|\mathbf{v}_i^{\text{(MRT-ZF)}}\|_{2}^{2}=(\sqrt{x_{i}}\mathbf{r}_{i}^{\text{(MRT)}}+\sqrt{y_{i}}\mathbf{r}_{i}^{\text{(ZF)}})^{H}(\sqrt{x_{i}}\mathbf{r}_{i}^{\text{(MRT)}}+\sqrt{y_{i}}\mathbf{r}_{i}^{\text{(ZF)}})\label{eq:powerMRTZF}\\
 & =x_{i}p_{x,i}+y_{i}p_{y,i}+s_{i}p_{s,i},\nonumber
\end{align}

\noindent where $p_{x,i}=\|\mathbf{r}_{i}^{\text{(MRT)}}\|^2=Q_{i,i}$,
$p_{y,i}=\|\mathbf{r}_{i}^{\text{(ZF)}}\|^2=q_{i}$ and
$p_{s,i}=(\mathbf{r}_{i}^{\text{(ZF)}})^{H}(\mathbf{r}_{i}^{\text{(MRT)}})+(\mathbf{r}_{i}^{\text{(MRT)}})^{H}(\mathbf{r}_{i}^{\text{(ZF)}})=2q_{i}$.
Substituting Eqs. \eqref{eq:chGainsMRTZF2} and \eqref{eq:powerMRTZF}
into problem \eqref{initFormulation} yields:
\begin{eqnarray}
{\displaystyle \min_{\mathbf{x},\mathbf{y},\mathbf{s},\boldsymbol \rho}} &  & {\displaystyle \sum_{i=1}^{K}}(x_{i}Q_{i,i}+y_{i}q_{i}+2s_{i}q_{i}) \label{eq:formulationMRTZF1} \\
x_{i}Q_{i,i}^{2}+y_{i}q_{i}^{2}+2s_{i}Q_{i,i}q_{i} & \ge & \gamma_{i}\left(\sigma^{2}+{\displaystyle \sum_{j\neq i}}x_{j}Q_{i,j}^{2}\right)+\frac{\gamma_{i}\sigma_{C}^{2}}{\rho_{i}},\,\forall i \nonumber \\
(1-\rho_{i})P_{i}^{r} & \geq & \lambda_{i},\,\forall i \nonumber \\
{\displaystyle \sum_{j=1}^{K}} x_{j}Q_{i,j}^{2} + y_{i}q_{i}^{2} + 2s_{i}Q_{i,i}q_{i}+\sigma^{2} & = & P_{i}^{r},\,\forall i \nonumber \\
s_{i} & = & \sqrt{x_{i}y_{i}},\,\forall i \nonumber \\
0\le\rho_{i}\le1, &  & \, x_{i}\ge0,\, y_{i}\ge0,\, s_{i}\ge0,\,\forall i. \nonumber
\end{eqnarray}

Formulation \eqref{eq:formulationMRTZF1} is not convex due to the constraint
$s_{i}=\sqrt{x_{i}y_{i}}$. Convexification can be achieved by relaxing
this constraint into $s_{i}\le\sqrt{x_{i}y_{i}}$ or $s_{i}^{2}\le x_{i}y_{i}$,
which is a convex second-order cone as $x_{i}\ge0$ and $y_{i}\ge0$.
The relaxed formulation can be easily converted into the following
convex SOCP:
\vspace{-0.3cm}
\begin{eqnarray}
{\displaystyle \min_{\mathbf{x},\mathbf{y},\mathbf{s},\boldsymbol \rho}} &  & {\displaystyle \sum_{i=1}^{K}}(x_{i}Q_{i,i}+y_{i}q_{i}+2s_{i}q_{i}) \label{eq:formulationMRTZF2} \\
z_{i}\rho_{i} & \ge & \left(\sqrt{\gamma_{i}\sigma_{C}^{2}}\right)^{2},\,\forall i \nonumber \\
(1-\rho_{i})P_{i}^{r} & \geq & \left(\sqrt{\lambda_{i}}\right)^{2},\,\forall i \nonumber \\
{\displaystyle \sum_{j=1}^{K}} x_{j} Q_{i,j}^{2} + y_{i}q_{i}^{2}+2s_{i}Q_{i,i}q_{i}+\sigma^{2} & = & P_{i}^{r},\,\forall i \nonumber \\
x_{i}y_{i} & \ge & s_{i}^{2},\,\forall i \nonumber \\
z_{i}+\gamma_{i}\sigma^{2} + \gamma_{i}{\displaystyle \sum_{j\neq i}}x_{j}Q_{i,j}^{2} & = & x_{i}Q_{i,i}^{2}+y_{i}q_{i}^{2}+2s_{i}Q_{i,i}q_{i},\,\forall i \nonumber \\
0\le\rho_{i}\le1,\, z_{i}\ge0, &  & \, x_{i}\ge0,\, y_{i}\ge0,\, s_{i}\ge0,\,\forall i.\nonumber
\end{eqnarray}
Because problem \eqref{eq:formulationMRTZF2} provides a lower bound approximation to \eqref{eq:formulationMRTZF1}, the solution obtained may not be valid for problem \eqref{initFormulation}. To deal with this issue, we use the decision vectors $\mathbf{x}$ and $\mathbf{y}$ obtained from the solution of \eqref{eq:formulationMRTZF2} to construct fixed, normalized weights $\mathbf{w}_i^{\text{(MRT-ZF)}}$ according to Eq. \eqref{eq:weightsMRTZF}, and then solve problem \eqref{eq:SOCP} to obtain a valid solution to \eqref{initFormulation}. In case the latter provides an infeasible solution, ZF beamforming is employed.

\vspace{-0.3cm}
\subsection{Implementation issues}
Most of the proposed schemes (except ZF beamforming) refer to global optimization problems that involve the centralized knowledge of all the downlink channels. A potential implementation requires a central processing unit/controller, which collects all the downlink channels from the transmitters and then communicates the optimal parameters (beamforming vectors, transmit powers, power split factors) to the system; similar centralized implementations have been proposed in \cite{HUA} for basic MISO interference channels. Although this solution corresponds to a high complexity and signaling overhead, modern cellular communication systems introduce base-station cooperation and provide a centralized back-haul network for sophisticated precoding; in our case, we have a coordination at the beamforming level and transmit signals remain locally known at the transmitters.  On the other hand, the purpose of this work is to introduce a new network structure (MISO interference channel with QoS and EH constraints) and study its optimal performance in terms of total energy consumption; implementations issues are beyond the scope of this work. The formulated optimization problems provide useful theoretical bounds and serve as guidelines for the evaluation of practical (distributed) implementations. The design of distributed algorithms that solve the optimization problems based on a local channel knowledge at each transmitter is an interesting problem for future work \cite{QIU}.

\vspace{-0.2cm}
\section{Numerical Results}

We evaluate the performance of the developed algorithms by solving several randomly generated problems for different parameter configurations. All problem instances follow the system model in Section \ref{system_model} with $\sigma^2=\sigma_C^2 =-40$ dBm. We assume that the attenuation from the sources to the receivers is $50$ dB for the direct channels and $\delta\cdot (50 \text{dB})$ with $\delta>0$ for the indirect channels\footnote{It corresponds to a symmetric topology with $\sigma_{ch_{i,i}}^2=10^{-5}$ and $\sigma_{ch_{i,j}}^2=10^{-5}/\delta$ for $i\neq j$, with $i=1,\ldots,K$ and $j=1,\ldots,N$.} (a simple method to ensure that the direct links are stronger than the interference links \cite{TAS}). The channel vectors are randomly generated from independent and identically distributed Rayleigh fading with average power as specified from the attenuation of the particular channel. For simplicity, the detection and energy harvesting thresholds are assumed to be equal for all users ($\gamma_i=\gamma$ and $\lambda_i=\lambda, \forall i$) and that the conversion efficiency is equal to $\zeta_i=1$. Mathematical modeling and solution of the SOCP formulations for fixed beamforming was done using the Gurobi optimization solver \cite{gurobi}, while the modeling and solution of the SDP relaxation for variable beamforming weights was performed using CVX \cite{cvx}.

\vspace{-0.3cm}
\subsection{Beamforming schemes performance comparison}

\begin{table}
    \centering
	    \caption{Infeasibility and optimality results for different parameter configurations for $\delta=20$.}
	    \vspace{-0.3cm}
	    \label{tab:resOI}
        \begin{tabular}{|c| c| c| c |c| c |c| c|c|c|c|}
            \hline
            \multicolumn{3}{|c|}{\textbf{Parameters}} & \multicolumn{3}{c|}{\textbf{Infeasibility}} & \multicolumn{5}{c|}{\textbf{Optimality}} \\
            \hline
            $K$ &   $\gamma$(dB)    &   $\lambda$ (dBm) &  $MRT$    &   $LP_{\rho}$ &   $RZF$ & $\frac{f_{MRT}}{f_{Opt}}$ & $\frac{f_{LP_{\rho}}}{f_{Opt}}$ & $\frac{f_{ZF}}{f_{Opt}}$ & $\frac{f_{RZF}}{f_{Opt}}$    & $\frac{f_{MRT-ZF}}{f_{Opt}}$         \\
            \hline
            \hline
                        2   &   20  &   -40 &   70  &   70  &   84  &   3.4863  &   5.2258  &   5.587   &   12.5926 &   1.0013  \\ \hline
                        2   &   20  &   -30 &   71  &   71  &   82  &   3.2814  &   4.8858  &   1.7049  &   3.2956  &   1.0018  \\ \hline
                        2   &   20  &   -20 &   76  &   76  &   87  &   2.3974  &   3.2925  &   2.6637  &   9.7787  &   1.0012  \\ \hline
                        2   &   20  &   -10 &   77  &   77  &   87  &   1.9305  &   2.7745  &   2.4875  &   3.3039  &   1.0007  \\ \hline
                        2   &   10  &   -40 &   2   &   2   &   20  &   1.2156  &   1.7927  &   4.1427  &   4.0098  &   1.0033  \\ \hline
                        2   &   10  &   -30 &   0   &   0   &   9   &   1.2496  &   1.5891  &   4.6628  &   2.9024  &   1.0028  \\ \hline
                        2   &   10  &   -20 &   3   &   3   &   16  &   1.0194  &   1.7411  &   5.9004  &   7.3744  &   1.0057  \\ \hline
                        2   &   10  &   -10 &   2   &   2   &   22  &   1.0655  &   2.0582  &   3.321   &   2.7239  &   1.0074  \\ \hline
                        4   &   20  &   -40 &   100 &   100 &   100 &   -   &   -   &   4.8337  &   -   &   1.7032  \\ \hline
                        4   &   20  &   -30 &   100 &   100 &   100 &   -   &   -   &   8.2793  &   -   &   1.6311  \\ \hline
                        4   &   20  &   -20 &   100 &   100 &   100 &   -   &   -   &   4.3944  &   -   &   1.5453  \\ \hline
                        4   &   20  &   -10 &   100 &   100 &   100 &   -   &   -   &   11.1585 &   -   &   1.7093  \\ \hline
                        4   &   10  &   -40 &   1   &   1   &   67  &   1.2929  &   1.9129  &   14.0505 &   37.641  &   1.1038  \\ \hline
                        4   &   10  &   -30 &   1   &   1   &   67  &   1.299   &   1.6987  &   9.3093  &   13.7661 &   1.0884  \\ \hline
                        4   &   10  &   -20 &   1   &   1   &   69  &   1.1137  &   1.7666  &   21.0165 &   41.9527 &   1.0261  \\ \hline
                        4   &   10  &   -10 &   1   &   1   &   66  &   1.0235  &   1.9934  &   41.7677 &   5.3096  &   1.0141  \\ \hline
                        8   &   20  &   -40 &   100 &   100 &   100 &   -   &   -   &   11.0335 &   -   &   2.879   \\ \hline
                        8   &   20  &   -30 &   100 &   100 &   100 &   -   &   -   &   9.9038  &   -   &   2.9878  \\ \hline
                        8   &   20  &   -20 &   100 &   100 &   100 &   -   &   -   &   14.6228 &   -   &   2.8204  \\ \hline
                        8   &   20  &   -10 &   100 &   100 &   100 &   -   &   -   &   15.899  &   -   &   2.9969  \\ \hline
                        8   &   10  &   -40 &   0   &   0   &   100 &   1.3704  &   2.0311  &   42.8582 &   -   &   1.2174  \\ \hline
                        8   &   10  &   -30 &   0   &   0   &   100 &   1.2703  &   1.6751  &   31.3244 &   -   &   1.1707  \\ \hline
                        8   &   10  &   -20 &   0   &   0   &   100 &   1.0439  &   1.6402  &   259.085 &   -   &   1.0306  \\ \hline
                        8   &   10  &   -10 &   0   &   0   &   100 &   1.0057  &   1.9607  &   48.9879 &   -   &   1.0057  \\ \hline
        \end{tabular}
\end{table}

To examine the relative performance, in terms of  optimality and infeasibility, of the beamforming approaches we solved several problems instances for different parameter configurations with $\delta=5$. The results are summarized on Table \ref{tab:resOI}, where each table entry is the average of $100$ randomly generated instances. In the table, $LP_\rho$ corresponds to the solution of the linear program obtained when the power splitting parameters are fixed and equal to $\rho_i=0.5$, $\forall i$, while ZF, MRT, RZF and MRT-ZF correspond to solutions associated with the particular beamforming schemes. Note that all optimality results are illustrated relative to the optimal solution ($f_{Opt}$) of \eqref{initFormulation}. This was obtained from the solution of the SDP problem (\ref{eqn:prob:sdp2}), as it provided rank-1 solutions in all instances considered. Note also that the hyphen symbol ``-'' in the table is used when no feasible solutions where obtained from a particular beamforming approach to be able to present optimality results.

Columns 4-6 show the percentage of infeasible instances for the different beamforming approaches. As expected, ZF, MRT-ZF and optimal beamforming exhibited no infeasible problem instances and are omitted from the table. Beamforming schemes MRT,  $\text{LP}_\rho$ and RZF exhibit a large number of infeasible solutions, especially when the SINR threshold is large ($\gamma=20$ dB). Comparatively, MRT and $\text{LP}_\rho$ exhibit exactly the same number of infeasible problem for all cases considered which indicates the structure of the weight vectors is more important than the power splitting parameter. Additionally, these schemes produce considerably more feasible problems compared to the RZF approach with $\eta_i = 1$. It seems the feasibility solely depends on the SINR but not the EH constraints.

Columns 7-11 indicate the optimality performance of different beamforming schemes with respect to the optimal solution. Comparing the optimality results of MRT and $\text{LP}_\rho$, it is easy to see that when the power splitting parameters are fixed ($\text{LP}_\rho$ approach) the required power is 50\%-100\% more than the case that these parameters are optimally selected (MRT approach). MRT also significantly outperforms RZF which indicates that the latter is not appropriate for the particular problem, at least for the considered scenaria. MRT also outperforms ZF in almost all cases for feasible MRT instances, because it produces strong cross-interference which facilitates the EH constraints. The above results demonstrate that from the fixed weight schemes, MRT and ZF are the best in terms of optimallity and feasibility respectively. By combining these schemes, MRT-ZF always produces feasible solutions and exhibits excellent performance, especially for small $K$. For $K=8$, the performance of MRT-ZF worsens but still for some parameter combinations its performance is excellent. It appears from the table that the most important parameter for the relative performance between MRT-ZF and optimal beamforming is the SINR threshold; when the SINR threshold is 20dB the relative performance of MRT-ZF compare to optimal beamforming is much worse than when SINR is $10$ dB. Regarding the optimal beamforming scheme, the SDP relaxation algorithm proposed in section \ref{sec:optBeamforming} produced rank-1 solutions in all simulations, indicating that its solution is optimal for all considered instances.

\begin{figure}[t]
\centering
\includegraphics[width=0.5\linewidth]{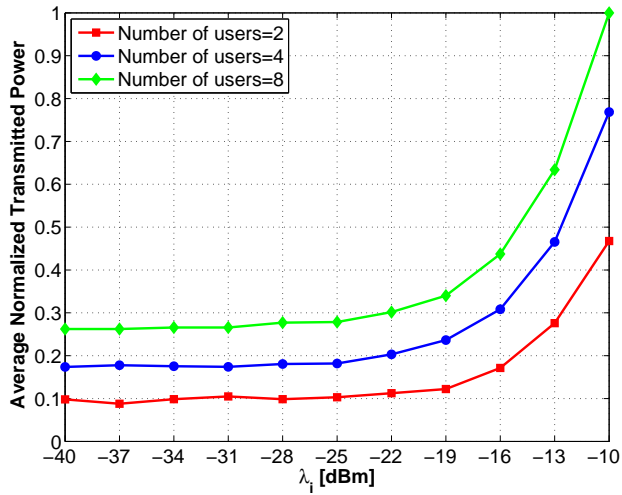}
\vspace{-0.8cm}
\caption{Total transmitted power versus EH threshold; $\gamma=20$ dB and $\delta=5$.}
\label{fig:powerVSeh}%
\end{figure}

\begin{figure}[t]
\centering
\includegraphics[width=0.5\linewidth]{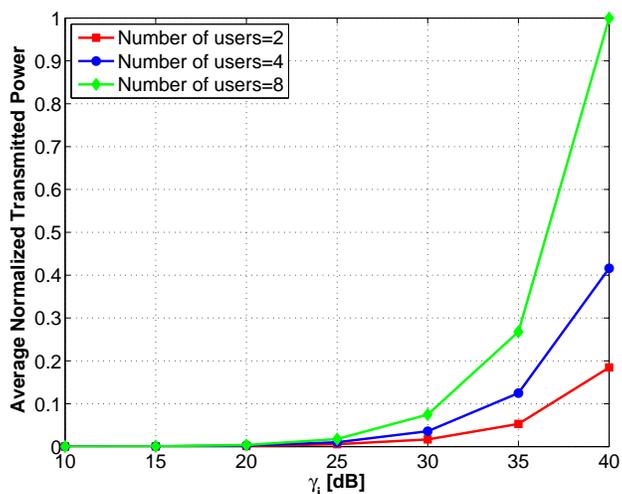}
\vspace{-0.8cm}
\caption{Total transmitted power versus SINR threshold; $\lambda=-30$ dBm and $\delta=5$.}
\label{fig:powerVSsnr}
\end{figure}

\begin{figure}[t]
\centering
\includegraphics[width=0.5\linewidth]{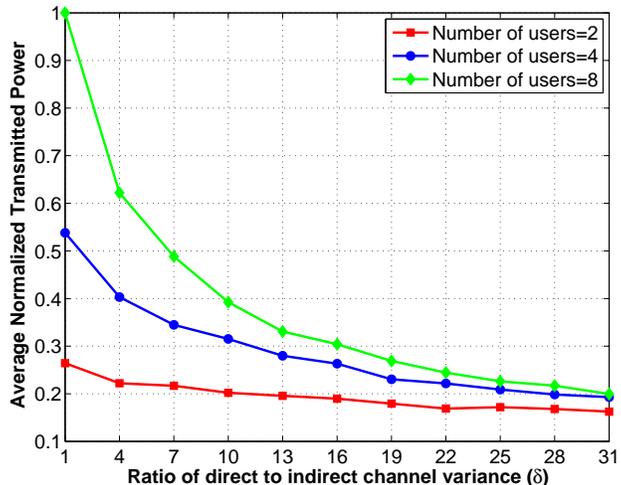}
\vspace{-0.8cm}
\caption{Total transmitted power versus the ratio of direct to indirect channel variance; $\lambda=-30$ dBm and $\gamma=20$ dB.}
\label{fig:powerVSdss}%
\end{figure}

\vspace{-0.1cm}
\subsection{Parameter effect}

Apart from the performance of the developed algorithms for the solution of problem \eqref{initFormulation}, the effect of parameters $\gamma_i$, $\lambda_i$ and $\delta$ has also been investigated as illustrated in Figs. \ref{fig:powerVSeh},  \ref{fig:powerVSsnr} and \ref{fig:powerVSdss} respectively. Each point in the figures is the average value of the total instantaneous transmit power calculated over $1000$ problem instances for a specific parameter combination; for comparison purposes these values have been normalized with respect to the maximum average value. As expected, smaller detection and energy harvesting thresholds result in lower transmit power requirements because the constraints are easier to satisfy. An interesting observation regarding Figs. \ref{fig:powerVSeh} and \ref{fig:powerVSsnr} is that there is a flat region of almost constant required power for different values of $\lambda_i$ and $\gamma_i$ respectively. This illustrates that the effect of the varying constraint is negligible below a certain threshold because of the power required to compute the non-varying threshold. For example, Fig. \ref{fig:powerVSeh} indicates that the transmitted power needed to satisfy a QoS threshold of $20$ dB, is enough to harvest at least $-22$ dBm of power.   

Regarding the effect of parameter $\delta$, Fig. \ref{fig:powerVSdss} demonstrates that when the ratio of direct to indirect channel variance increases the required transmission power decreases; this implies that the SINR constraints play a more important role than the EH constraints. The reason is that by increasing $\delta$ the SINR constraints can be satisfied easier, while the total power received at each destination is small due to weak cross-interference. Finally, notice that the average transmitted power increases for increasing number of users $K$.

\subsection{Interference Exploitation Benefit}

An interesting implication that arises from problem \eqref{initFormulation}, is that there is a clear trade-off between eliminating or allowing interference towards the satisfaction of the QoS and EH constraints. This implication diverges from the traditional beamforming design philosophy of only eliminating interference. To illustrate the benefit of exploiting interference in the context of RF-EH we compare the performance between ZF and optimal beamforming schemes. On the one hand, the ZF scheme cancels out interference between users minimizing the power required to satisfy the SINR constraints, but on the other hand, ZF fails to exploit interference to satisfy the RF-EH constraints; this must be accomplished solely from the associated transmitter. 

\begin{figure}[t]
\centering
\includegraphics[width=0.6\linewidth]{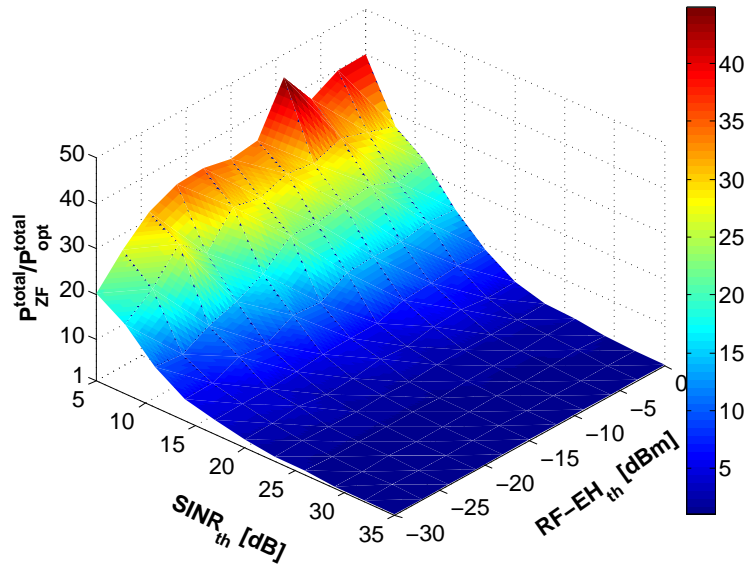}
\vspace{-0.8cm}
\caption{Average ratio between the total required power for the ZF beamforming scheme and the optimal scheme for varying SINR and EH thresholds for $K=8$ and $\delta = 5$.}
\label{fig:interferenceExploitation}%
\end{figure}

Fig. \ref{fig:interferenceExploitation} depicts the average ratio between the total required power when ZF beamforming is used and when the optimal solution is employed (obtained via SDR) for various values of the SINR threshold and the RF-EH threshold when $K=8$. For each combination of values, 100 problem instances were generated and solved to compute average ratio values. As can be seen, by exploiting interference we can significantly reduce the total transmitted power especially for low SINR values (up to 45 times in the considered scenarios). The reason for this is that when the SINR threshold is low, there is room to increase interference, which is beneficial for the RF-EH constraint, without violating the SINR threshold. On the other hand, having to satisfy large SINR thresholds is difficult and requires almost full cancellation of the interference; hence, the solutions obtained from the ZF beamformer are almost optimal. The benefits of interference exploitation can also be seen with respect to the RF-EH constraints: when the RF-EH threshold increases, the ZF/optimal power ratio increases because the optimal scheme manages interference better. 

\subsection{Trade-off between beamforming schemes}

\begin{figure}[t]
\centering

\subfigure[Absolute execution time.]{
  \includegraphics[width=0.47\linewidth]{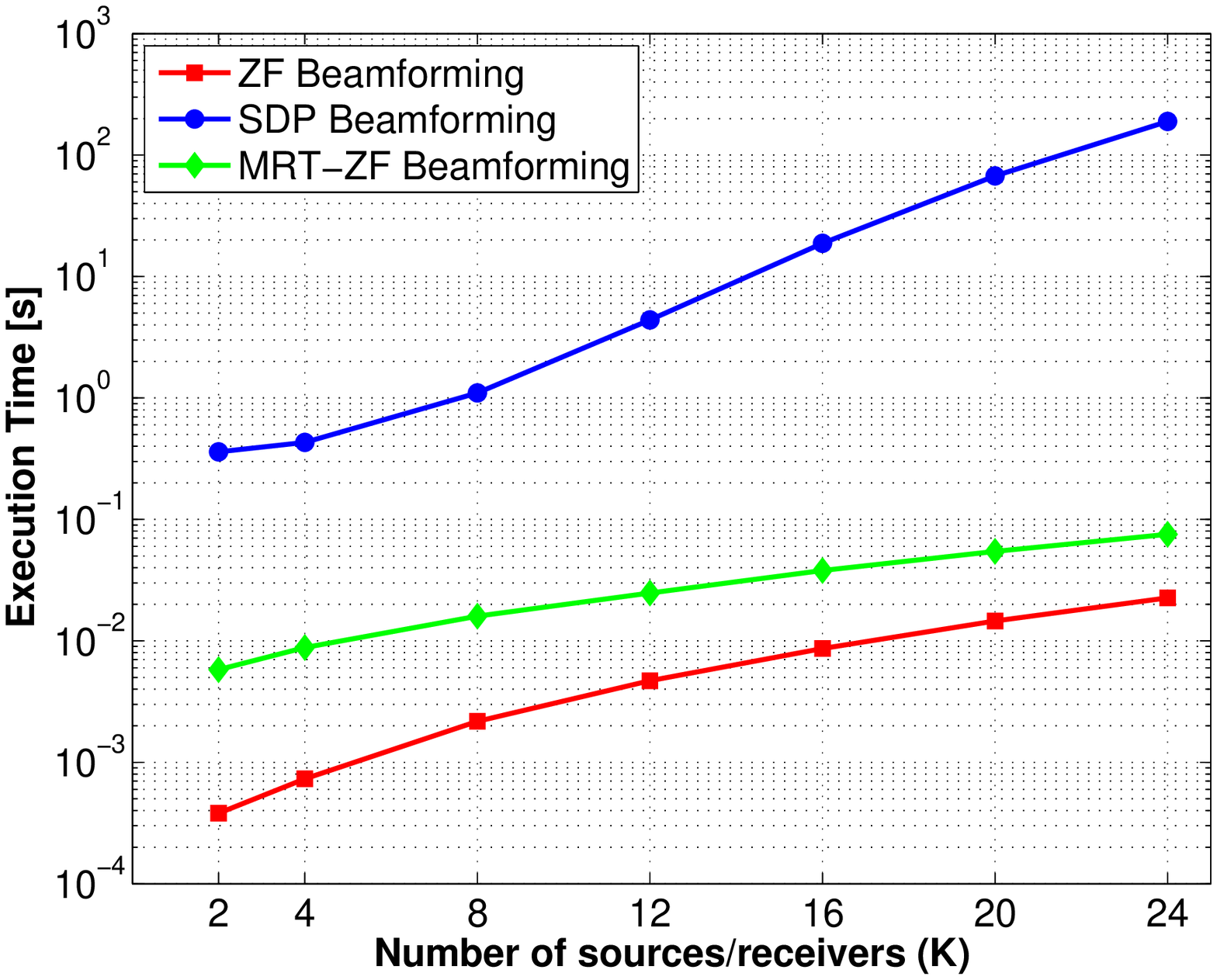}
  \label{fig:execTimes}
 }
 \subfigure[Relative execution time.]{
 \includegraphics[width=0.47\linewidth]{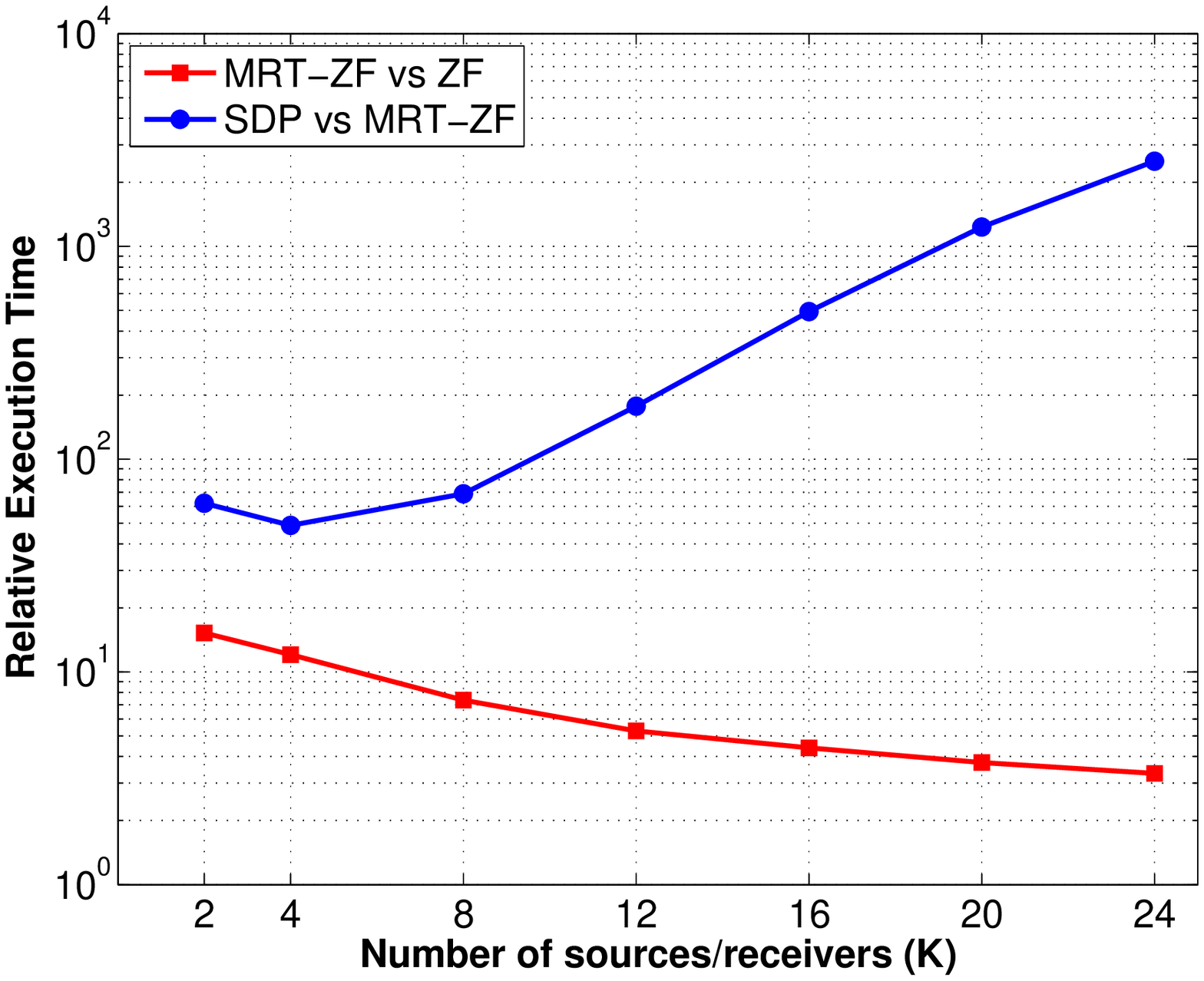}
   \label{fig:execTimesRelative}
 }

\label{myfigure}
\vspace{-0.3cm}
\caption{Time needed to obtain ZF, MRT-ZF and optimal(SDP) beamforming weights for different problem sizes.}
\end{figure}

Given the fact that the SDP beamforming scheme provides optimal solutions to all instances considered, why should one consider a different beamforming scheme? One important issue that needs to be taken into consideration before deciding which beamforming scheme to use is related to the imposed time-constraints regarding the solution computation. In other words, how much time is available to compute appropriate beamforming weights? To answer the question one needs to consider the complexity of computing the solution for the different beamforming schemes proposed. In this section we empirically investigate the computational complexity of the ZF, MRT-ZF and optimal (SDP) schemes by observing their execution times for different different problems.

It is well-known that the complexity of solving SDP problems is significantly higher than that of SOCP problems. Hence it is expected that the execution time for optimal beamforming will be higher compared to MRT-ZF. On the other hand, for ZF beamforming we need to compute the pseudo-inverse of different matrices and also to find the closed-form  solution for $P_i$ and $\rho_i$.

Figure \ref{fig:execTimes}, presents the average execution time (in logarithmic scale) of these beamforming schemes for $20$ problem instances for a particular parameter combination and different values of $K$\footnote{All problems were executed on a desktop computer with an Intel Core 2 Duo CPU at 3GHz(E8400) and 3GB of DDR2 RAM.}. The figure demonstrates that the average execution times of ZF and MRT-ZF are below $0.025$s and $0.08$s for all $K$ considered, while the SDP approach employed to yield the optimal solution requires about $190$s for $K=24$. Hence it is clear that as $K$ increases the use of optimal beamforming becomes prohibitive, especially for $K\ge12$.  Regarding their relative execution times depicted on Fig. \ref{fig:execTimesRelative}, as $K$ increases the relative execution time of MRT-ZF compared to ZF decreases, while the execution time for optimal beamforming grows significantly relative to the MRT-ZF time. Interestingly, for $K=24$, the relative execution time of MRT-ZF compared to ZF is only $3.3$, while the relative execution time of optimal beamforming compared to MRT-ZF is more that $2500$ times.

Based on these results, it is clear that although the SDP approach provides the best performance in terms of optimality, MRT-ZF provides the best trade-off between optimality/feasibility performance and execution time, especially as $K$ grows large.

\section{Conclusion}
This paper has dealt with the RF-EH power splitting technique for a MISO interference channel with QoS and EH constraints. The minimum required energy has been formulated via an optimization problem for constant or variable beamforming weights at the sources.  Solution algorithms have been investigated for three standard MISO beamforming designs, ZF, RZF and MRT, for a hybrid beamforming scheme, MRT-ZF, combining the MRT and ZF beamformers, as well as for variable weight beamforming that provides the optimal solution. For ZF, a closed-form always feasible solution has been derived, while for MRT and RZF a convex SOCP program has been obtained and solved to optimality. For MRT-ZF, an algorithm has been developed which requires the solution of two SOCP problems; MRT-ZF significantly outperformed all other fixed weight beamforming schemes. For the solution of the optimal beamforming problem an approximate SDP formulation has been developed and it was theoretically proved that it provides optimal results for two and three users; in practice the algorithm always exhibited optimal performance. Finally, the computational complexity of the different schemes has been examined indicating that SDP is prohibitive for medium scale systems (e.g. with 20 users), while MRT-ZF provides the best trade-off between computational complexity and optimality. An extension of this work is to apply the investigated schemes for scenarios with limited channel feedback and imperfect CSI at the transmitters.

\appendices

\section{Proof of Theorem 2}

 In optimization problem ${\cal ZF}$, the $i$th set of constraints
\eqref{eq:ZF2}-\eqref{eq:ZF4} are decoupled as the variables $P_{i}$, $\rho_{i}$ do not appear
in other constraints; also, the variables in the objective function
are in separable form. Hence, problem ${\cal ZF}$ can be decomposed
into $K$ independent problems ${\cal SZF}_{i}$, defined as:
\begin{eqnarray}
{\cal SZF}_{i}: & {\displaystyle \min_{P_{i},\rho_{i}}} & P_{i}\nonumber \\
\frac{\rho_{i}G_{i,i}P_{i}}{\rho_{i}\sigma^{2}+\sigma_{C}^{2}} & \ge & \gamma_{i},\label{eq:szf_constr1}\\
(1-\rho_{i})(G_{i,i}P_{i}+\sigma^{2}) & \ge & \lambda_{i}\label{eq:szf_constr2}\\
P_{i}\ge0 & , & 0\le\rho_{i}\le1.\nonumber
\end{eqnarray}

\vspace{-0.1cm}
The optimal solution of ${\cal ZF}$ is equal to the sum of the optimal solutions of ${\cal SZF}_{i}$. Note that in the considered problem, $\rho_{i}\in(0,1)$, otherwise the problem constraints are not satisfied for $\gamma_{i},\lambda_{i}>0$.

Assuming that:
\vspace{-0.1cm}
\begin{eqnarray}
x_{i} = G_{i,i}P_{i}+\sigma^{2}\ge\sigma^{2},\;\; \alpha_i = (\gamma_{i}+1)\sigma^{2},\;\; \beta_i =\gamma_{i}\sigma_{C}^{2}, \nonumber 
\end{eqnarray}

\noindent constraints (\ref{eq:szf_constr1}) and (\ref{eq:szf_constr2}) can
be written as:
\begin{eqnarray}
\rho_{i} & \ge & \frac{\beta_i}{x_{i}-\alpha_i}<1,\label{eq:szf_constr1ab}\\
\rho_{i} & \le & 1-\frac{\lambda_{i}}{x_{i}}>0.\label{eq:szf_constr2ab}
\end{eqnarray}

This implies that $x_{i}>\alpha_i+\beta_i$ and $x_{i}>\lambda_{i}$,
otherwise $\rho_{i}\notin(0,1)$; in other words $x_i$ must be greater than $\max(\lambda_{i},\alpha_i+\beta_i)$.
Problem ${\cal SZF}_{i}$ requires at least one of the constraints
to be binding, otherwise the value of $P_{i}$ can further be reduced.
Hence, at least one of the following two equalities must be true at
the optimal solution.
\begin{equation}
\rho_{i}=\frac{\beta_i}{x_{i}-\alpha_i},\label{eq:szf_constr1_binding}
\end{equation}
\begin{equation}
\rho_{i}=1-\frac{\lambda_{i}}{x_{i}}.\label{eq:szf_constr2_binding}
\end{equation}

Let us assume that (\ref{eq:szf_constr1_binding}) is true. Substituting this equation into (\ref{eq:szf_constr2ab}) and rearranging the terms yields:
\begin{equation}
x_{i}^{2}-(\alpha_i+\beta_i+\lambda_{i})x_{i}+\alpha_i\lambda_{i}\ge0.  \label{eq:szf_ineq1}
\end{equation}

Interestingly, we obtain (\ref{eq:szf_ineq1}) even if we follow the same procedure for the other
binding constraint. This implies that binding any of the two equations will yield the same solution.
Because $x_{i}$ is a monotonically increasing function of $P_{i}$,
the optimal solution to problem ${\cal SZF}_{i}$ is the smallest
value of $x_{i}$ which satisfies (\ref{eq:szf_ineq1}) and $x_{i}>\max(\lambda_{i},\alpha_i+\beta_i)$.
It can be easily verified that the
discriminant $\Delta$ of the quadratic expression in \eqref{eq:szf_ineq1} is always positive, which implies that there are
two distinct real solutions $x_{1}$ and $x_{2}$, and that the feasible region of \eqref{eq:szf_ineq1} is $x\le x_{1}$ and $x\ge x_{2}$,
where:
\begin{align}
x_{1} &=\frac{1}{2}\left(\alpha_i+\beta_i+\lambda_{i}-\sqrt{(\alpha_i+\beta_i+\lambda_{i})^{2}-4\alpha_i\lambda_{i}}\right), \label{eq:szf_sol1}\\
x_{2} &=\frac{1}{2}\left(\alpha_i+\beta_i+\lambda_{i}+\sqrt{(\alpha_i+\beta_i+\lambda_{i})^{2}-4\alpha_i\lambda_{i}}\right). \label{eq:szf_sol2}
\end{align}

It can be easily    shown that $x_{1}<\max(\lambda_{i},\alpha_i+\beta_i)$, because
$\Delta>0$, and $\alpha_i+\beta_i+\lambda_{i}\le2\max(\lambda_{i},\alpha_i+\beta_i)$.
Next we show that $\max(\lambda_{i},\alpha_i+\beta_i)<x_{2}$.

If we assume that $\lambda_{i}\ge\alpha_i+\beta_i$ we need to show that:
\begin{align*}
&\alpha_i+\beta_i+\lambda_{i}+\sqrt{(\alpha_i+\beta_i+\lambda_{i})^{2}-4\alpha_i\lambda_{i}} > 2\lambda_{i}\Rightarrow\\
&(\alpha_i+\beta_i+\lambda_{i})^{2}-4\alpha_i\lambda_{i} > (\alpha_i+\beta_i-\lambda_{i})^{2}\Rightarrow\\
&4\beta_i\lambda_{i} > 0,
\end{align*}
\noindent where the latter inequality is true because $\beta_i,\lambda_{i}>0$.
In a similar manner we can easily show that $\max(\lambda_{i},\alpha_i+\beta_i)<x_{2}$
when $\lambda_{i}\le\alpha_i+\beta_i$. Hence, we have shown that $x_{1}<\max(\lambda_{i},\alpha_i+\beta_i)<x_{2}$,
which implies that the optimal solution is $x_{i}^{*}=x_{2}$. In
addition, it can be derived that $\rho_{i}^{*}=\frac{\beta_i}{x_{i}^{*}-\alpha_i}=1-\frac{\lambda_{i}}{x_{i}^{*}}$
which implies that both constraints (\ref{eq:szf_constr1ab}) and
(\ref{eq:szf_constr2ab}) are binding at the solution. Having derived
$x_{i}^{*}$ and $\rho_{i}^{*}$, the optimal power value is $P_{i}^{*}=\frac{1}{G_{i,i}}(x_{i}^{*}-\sigma^{2})$ which
completes the proof.

\end{document}